\documentclass[12pt,a4paper]{article}
\usepackage{fullpage}
\usepackage{amssymb,amsmath,stmaryrd,amsthm}
\usepackage[mathscr]{eucal}
\usepackage{pstricks,pst-node,pst-text,pst-3d}

\theoremstyle{definition}

\newtheorem{theorem}{Theorem}
\newtheorem{proposition}{Proposition}
\newtheorem{lemma}{Lemma}
\newtheorem{corollary}{Corollary}

\theoremstyle{remark}
\newtheorem{remark}{Remark}

\def \RR {\mathbb{R}}
\def \1{{\mathbf{1}}}
\def \0{{\mathbf{0}}}
\def \C {\mathscr{C}}

\def \NC {\mathscr{NC}}
\def \SNC {\mathscr{SNC}}

\begin{document}

\title{\bf\Large Exact bounds of the M\"obius inverse of monotone set functions}

\author{Michel GRABISCH\thanks{Corresponding author. Paris School of Economics,
          University of Paris I,
         106-112, Bd. de l'H\^opital, 75013 Paris, France. Tel. (33)
        144-07-82-85, Fax (33)-144-07-83-01. Email:
         \texttt{michel.grabisch@univ-paris1.fr}}       \and Pedro
         MIRANDA\thanks{Complutense University of Madrid. Plaza de Ciencias, 3,
           28040 Madrid, Spain Email:
       \texttt{pmiranda@mat.ucm.es}}
}

\date{\today}

\maketitle

\begin{abstract}
We give the exact upper and lower bounds of the M\"obius inverse of monotone and
normalized set functions (a.k.a. normalized capacities) on a finite set of $n$
elements. We find that the absolute value of the bounds tend to
$\frac{4^{n/2}}{\sqrt{\pi n/2}}$ when $n$ is large. We establish also the exact bounds of
the interaction transform and Banzhaf interaction transform, as well as the
exact bounds of the M\"obius inverse for the subfamilies of $k$-additive normalized
capacities and $p$-symmetric normalized capacities.
\end{abstract}

\noindent {\bf Keywords:} M\"obius inverse, monotone set function, interaction

\noindent AMS Classification: 05, 06, 91

\section{Introduction}
The M\"obius function is a well-known tool in combinatorics and partially
ordered sets (see, e.g., \cite{aig79,lin97,rot64}). In the field of
decision theory, the M\"obius inverse of a monotone set function (called a
capacity) is a fundamental concept permitting to derive simple expressions of
nonadditive integrals and to analyze the core of capacities (set of probability
measures dominating a capacity) \cite{chja89}. Set functions can also be seen as
pseudo-Boolean functions, and it is well known that the M\"obius inverse
corresponds to the coefficients of the polynomial representation of a
pseudo-Boolean function.  In particular, monotone and normalized
  pseudo-Boolean functions correspond to semicoherent structure functions in
  reliability theory (see, e.g., Marichal and Mathonet \cite{mama13}, Marichal
  \cite{mar14}). 

Consider $N=\{1,\ldots,n\}$ and a monotone set function
$\mu:2^{N}\rightarrow[0,1]$ with the property $\mu(\varnothing)=0$ and
$\mu(N)=1$ (normalized capacity).  In optimization problems involving capacities
or monotone pseudo-Boolean functions (as in reliability) it is often
useful to know the bounds of the M\"obius inverse to use algorithmic
  methods (see Crama and Hammer \cite{crha11}, Chapter~13). This is the case
for example when dealing with $k$-additive measures, which are best represented
through their M\"obius inverse (see below); then, when solving optimization
problems like model fitting, algorithms usually need to fix an interval where
the searched values lay, and the upper and lower bounds are the natural limits
of these intervals. Surprisingly, although $\mu$ takes values in $[0,1]$, the
exact bounds of its M\"obius inverse grow rapidly with $n$, approximately in
$\frac{4^{n/2}}{\sqrt{\pi n/2}}$ when $n$ is large. The aim of the paper is to
establish this result, correcting wrong bounds obtained in a previous paper by
the authors \cite{migr99a}, and providing a complete proof of the result. We
extend this result to the interaction transform, another useful linear
invertible transform of set functions, and we consider also specific subclasses
of capacities, like $k$-additive and $p$-symmetric capacities.

\section{Preliminaries}
Let $N=\{1,\ldots,n\}$. A \textit{capacity} on $N$ is a set function
$\mu:2^{N}\rightarrow \RR$ satisfying $\mu(\varnothing)=0$ and monotonicity:
$A\subseteq B\subseteq N$ implies $\mu(A)\leqslant \mu(B)$. A capacity is
\textit{normalized} if in addition $\mu(N)=1$. We denote respectively by $\C(N)$
and $\NC(N)$ the set of capacities and normalized capacities on $N$. The set
$\NC(N)$ is a convex closed polytope, whose extreme points are all
$\{0,1\}$-valued normalized capacities (as the polytope of normalized capacities
is an order polytope, this result has been shown by Stanley \cite{sta86}. For a
direct proof, see \cite{rad98}). We denote by
$\NC_{0,1}(N)$ the set of all $\{0,1\}$-valued normalized capacities.

Consider a set function $\xi$ on $N$ such that $\xi(\varnothing)=0$. The
\textit{monotonic cover} of $\xi$ is the smallest capacity $\mu$ such that
$\mu\geqslant \xi$. We denote it by $\widehat{\xi}$, and it is given by
\begin{equation}\label{eq:mc}
\widehat{\xi}(A) = \max_{B\subseteq A}\xi(B)\qquad (A\subseteq N).
\end{equation}

Consider now a set function $\xi:2^N\rightarrow \RR$. The linear system
\begin{equation}\label{eq:1}
\xi(A) = \sum_{B\subseteq A}m(B) \qquad (A\in 2^N)
\end{equation}
has always a unique solution, known as the \textit{M\"obius inverse}
\cite{rot64}, and is given by
\begin{equation}\label{eq:2}
m(A) = \sum_{B\subseteq A}(-1)^{|A\setminus B|}\xi(B) \qquad (A\in 2^N).
\end{equation}
Since $m$ is also a set function, we view now the M\"obius inverse as a transform on
the set of set functions:
\[
m:\RR^{(2^N)}\rightarrow \RR^{(2^N)}; \xi\mapsto m^\xi\text{ given by (\ref{eq:2})}.
\]
We call $m$ the \textit{M\"obius transform} of $\xi $. Remark that it is a linear invertible
transform.

We introduce another linear invertible transform, which is useful in decision
making, called the \textit{interaction transform}. To this end we introduce the
derivative of a set function $\xi$. Let $i\in N$ and $A\subseteq N\setminus\{i\}$. The
\textit{derivative} of $\xi$ w.r.t. $i$ at $A$ is defined by $\Delta_i\xi(A) =
\xi(A\cup \{i\}) - \xi(A)$. Derivatives w.r.t. sets are defined recursively by
\[
\Delta_K\xi(A) = \Delta_{K\setminus \{i\}}(\Delta_{\{i\}}\xi(A)) \qquad
(|K|\geqslant 1)
\]
with $i\in K$, $\Delta_{\{i\}}\xi=\Delta_i\xi$, and
$\Delta_\varnothing\xi=\xi$. For $A\subseteq N\setminus K$, we obtain
\[
\Delta_K\xi(A) = \sum_{L\subseteq K}(-1)^{|K\setminus L|}\xi(A\cup L).
\]
 Also, observe that 
\begin{equation}\label{eq:md}
m^\xi(A) = \Delta_A\xi(\varnothing) \qquad (A\in 2^N).
\end{equation}
The interaction transform $I:\RR^{(2^N)}\rightarrow \RR^{(2^N)}$ computes a weighted
average of the derivatives:
\begin{equation}\label{eq:3}
I^\xi(A) = \sum_{B\subseteq N\setminus
  A}\frac{(n-b-a)!b!}{(n-a+1)!}\Delta_A\xi(B) \qquad (A\in 2^N),
\end{equation}
where $a=|A|, b=|B|$. Its expression through the M\"obius transform is much
simpler:
\begin{equation}\label{eq:4}
I^\xi(A) = \sum_{B\supseteq A}\frac{1}{b-a+1}m^\xi(B),
\end{equation}
while the inverse relation uses the Bernoulli numbers $B_k$:
\[
m^\xi(A) = \sum_{B\supseteq A}B_{a-b}I^\xi(B).
\]
(see \cite{degr96,grmaro99a} for details). Another related transform is the
\textit{Banzhaf interaction transform} $I_{\mathrm{B}}$, which is the
(unweighted) average of the derivatives:
\begin{equation}\label{eq:3b}
I^\xi_{\mathrm{B}}(A) = \frac{1}{2^{n-a}}\sum_{B\subseteq N\setminus
  A}\Delta_A\xi(B) \qquad (A\in 2^N).
\end{equation}

\medskip

Lastly, we introduce two specific families of normalized capacities. A
normalized capacity $\mu$ is said to be \textit{at most $k$-additive}
($1\leqslant k\leqslant n$) if $m^\mu(A)=0$ for every set $A\in 2^N$ such that
$|A|>k$ \cite{gra96f}. 1-additive capacities are ordinary additive capacities,
i.e., satisfying $\mu(A\cup B)=\mu(A)+\mu(B)$ for disjoint sets $A,B$. Note that
by (\ref{eq:4}), $m^\mu$ can be replaced by $I^\mu$ in the above definition.

We denote by $\NC^{\leqslant k}(N)$ the set of at most $k$-additive capacities
on $N$. It is a convex closed polytope (see \cite{micogi06} for a study of its properties).

Another family of interest is the family of $p$-symmetric capacities
\cite{migrgi02}. A capacity $\mu$ is \textit{symmetric} if $\mu(A)=\mu(B)$
whenever $|A|=|B|$.  We denote by $\SNC(N)$ the set of symmetric normalized
capacities. This notion can be generalized as follows. A nonempty subset
$A\subseteq N$ is a \textit{subset of indifference} for $\mu$ if for all
$B_1,B_2\subset A$ with $|B_1|=|B_2|$, we have $\mu(C\cup B_1)= \mu(C\cup B_2)$
for every $C\subseteq N\setminus A$. The \textit{basis} of the capacity is the
coarsest partition of $N$ into subsets of indifference. It always exists and is
unique \cite{migr03a}. Now, $\mu$ is \textit{$p$-symmetric} with respect to the
partition $\{ A_1,\ldots, A_p\} $ if this partition is its basis. Symmetric
games are therefore 1-symmetric games (with respect to the basis $\{ N\} $). We
denote by $\SNC^{\leqslant p}(A_1,\ldots, A_p)$ the set of normalized
capacities such that $A_1,\ldots, A_p$ are subsets of indifference. It is a
convex closed polytope (again, see \cite{micogi06} for a study of its properties).

Lastly, we mention a combinatorial result on the binomial coefficients:
\begin{equation}\label{eq:1.comb2}
\sum_{\ell=0}^k(-1)^\ell\binom{n}{\ell} = (-1)^k\binom{n-1}{k} \quad (k<n),
\end{equation}
for any positive integer $n$.

\section{Exact bounds of the M\"obius inverse}
We present in this section the main result of the paper.

\begin{theorem}\label{th:2.boundM}
For any normalized capacity $\mu$, its M\"obius transform satisfies for any
$A\subseteq N$, $|A|>1$:
\[
-\binom{|A|-1}{l'_{|A|}} \leqslant m^\mu(A)\leqslant \binom{|A|-1}{l_{|A|}},
\]
with
\begin{equation}\label{eq:boundm}
l_{|A|}  = 2\left\lfloor \frac{|A|}{4}\right\rfloor, \quad
l'_{|A|}  = 2\left\lfloor \frac{|A|-1}{4}\right\rfloor +1
\end{equation}
and for $|A|=1<n$:
\[
0\leqslant m^\mu(A) \leqslant 1,
\]
and $m^\mu(A)=1$ if $|A|=n=1$.
These upper and  lower bounds are attained by the normalized capacities
$\mu^*_A,\mu_{A*}$, respectively:
\[
\mu^*_A(B) = \begin{cases}
1, &\text{if } |A|-l_{|A|}\leqslant |B\cap A|\leqslant |A|\\ 0, & \text{otherwise}
\end{cases},\quad \mu_{A*}(B) = \begin{cases}
1, &\text{ if } |A|-l'_{|A|}\leqslant |B\cap A|\leqslant |A|\\ 0, & \text{otherwise}\end{cases}
\]
for any $B\subseteq N$.
\end{theorem}

We give in Table~\ref{tab:2.mobb} the first values of the bounds.
\begin{table}[htb]
\begin{tabular}{|l|rrrrrrrrrrrr|}\hline
$|A|$ & 1 & 2 & 3 & 4 & 5 & 6 & 7 & 8 & 9 & 10 & 11 & 12\\ \hline
u.b. of $m^\mu(A)$ & 1 & 1 & 1 & 3 & 6 & 10 & 15 & 35 & 70 & 126 & 210 & 462\\ \hline
l.b. of $m^\mu(A)$ & $1(0)$ & $-1$ & $-2$ & $-3$ & $-4$ & $-10$ & $-20$ & $-35$ & $-56$ &
$-126$ & $-252$ & $-462$ \\ \hline
\end{tabular}
\caption{Lower and upper bounds for the M\"obius transform of a normalized
  capacity}
\label{tab:2.mobb}
\end{table}
Using the well-known Stirling's approximation
$\binom{2n}{n}\simeq\frac{4^n}{\sqrt{\pi n}}$ for $n\rightarrow\infty$, we
deduce that
\[
-\frac{4^{\frac{n}{2}}}{\sqrt{\frac{\pi n}{2}}}\leqslant m^\mu(N)\leqslant
\frac{4^{\frac{n}{2}}}{\sqrt{\frac{\pi n}{2}}}
\]
when $n$ tends to infinity.

\begin{proof}
Let us prove the result for the upper bound when $A=N$.
 We consider the group $S_n$ of permutations on $N$. For any $\sigma\in S_n$
and any capacity $\mu\in \NC(N)$, we define the capacity $\sigma(\mu)\in \NC(N)$ by
$\sigma(\mu)(B) = \mu(\sigma^{-1}(B))$ for any $B\subseteq N$.

We observe that the target function $m^\mu(N)$ is invariant under
permutation. Indeed, 
\begin{align*}
m^{\sigma(\mu)}(N)  &= \sum_{B\subseteq N}(-1)^{n-|B|}\mu(\sigma^{-1}(B))\\
 & = \sum_{B'\subseteq N}(-1)^{n-|B'|}\mu(B') \qquad \text{(letting
}B'=\sigma^{-1}(B))\\
 & = m^\mu(N).
\end{align*}
For every set function $\mu$ on $N$, define its symmetric part $\mu^s =
\frac{1}{n!}\sum_{\sigma\in S_n}\sigma(\mu)$, which is a symmetric function. By
convexity of $\NC(N)$, if $\mu\in\NC(N)$, then so is $\mu^s$, and by linearity
of the M\"obius inverse, we have
\[
m^{\mu^s}(N) = \frac{1}{n!}\sum_{\sigma\in S_n}m^{\sigma(\mu)}(N) =
\frac{1}{n!}\sum_{\sigma\in S_n}m^\mu(N) = m^\mu(N).
\]
It is therefore sufficient to maximize $m^\mu(N)$ on the set of symmetric
normalized capacities $\SNC(N)$. But this set is also a convex polytope, whose
extreme points are the following $\{0,1\}$-valued capacities
$\mu_k$ defined by
\[
\mu_k(B) = 1 \text{ iff } |B|\geqslant n-k \qquad (k=0,\ldots, n-1).
\]
Indeed, if $\mu$ is symmetric, it can be written as a convex combination of
these capacities:
\[
\mu = \mu(\{1\})\mu_{n-1} + \sum_{k=2}^n\big(\mu(\{1,\ldots, k\}) -
\mu(\{1,\ldots, k-1\})\big)\mu_{n-k}
\]
It follows that the maximum of $m^\mu(N)$ is attained on one of these
capacities, say $\mu_k$. We compute
\begin{align}
m^{\mu_k}(N) & = \sum_{B\subseteq N}(-1)^{|N\setminus B|}\mu(B)
=\sum_{i=n-k}^n(-1)^{n-i}\binom{n}{i}\nonumber\\
 & = \sum_{i'=0}^k(-1)^{i'}\binom{n}{n-i'} = (-1)^{k}\binom{n-1}{k},\label{eq:boundm1}
\end{align}
where the third equality is obtained by letting $i'=n-i$ and the last one
follows from (\ref{eq:1.comb2}). 
Therefore $k$ must be even.  If $n-1$ is even, the maximum of $\binom{n-1}{k}$
for $k$ even is attained for $k=\frac{n-1}{2}$ if this is an even number,
otherwise $k=\frac{n-3}{2}$. If $n-1$ is odd, the maximum of $\binom{n-1}{k}$ is
reached for $k=\lceil\frac{n-1}{2}\rceil$ and $k-1=\lfloor\frac{n-1}{2}\rfloor$,
among which the even one must be chosen. As it can be checked (see
Table~\ref{tab:1} below), this amounts to taking
\[
k = 2\left\lfloor \frac{n}{4}\right\rfloor
\]
that is, $k=l_n$ as defined in (\ref{eq:boundm}), and we have defined the capacity
\[
\mu^*(B) = 1\text{ if } n-l_{n}\leqslant |B|\leqslant n,
\]
which is $\mu^*_N$ as defined in the theorem.

For establishing the upper bound of $m^\mu(A)$ for any $A\subset N$,
  remark that the value of $m^\mu(A)$ depends only on the subsets of $A$. It
  follows that applying the above result to the sublattice $2^A$, the set
  function $\xi^*_A$ defined on $2^N$ by
\[
\xi^*_A(B) = 1 \text{ if } B\subseteq A \text{ and }|A|-l_{|A|}\leqslant |B|\leqslant |A|,
\text{ and 0 otherwise},
\]
yields an optimal value for $m^\mu(A)$. It remains to turn this set function
into a capacity on $N$, without destroying optimality. This can be done since $\xi^*_A$
is monotone on $2^A$, so that taking the monotonic cover of $\xi^*_A$ by
(\ref{eq:mc}) yields an optimal capacity, given by
\[
\widehat{\xi^*_A}(B) = \max_{C\subseteq B}\xi^*_A(C) = 1\text{ if }
|A|-l_{|A|}\leqslant |B\cap
A|\leqslant |A|, \text{ and 0 otherwise},
\]
which is exactly $\mu^*_A$ as desired. Note however that this is not the only
optimal solution in general, since values of the capacty on the sublattice
$2^{N\setminus A}$ are irrelevant.

One can proceed in a similar way for the lower bound. In this case however, as
it can be checked, the
capacity must be equal to 1 on the $l'_n+1$ first lines of the lattice $2^N$,
with $l'_n=2\left\lfloor\frac{n-1}{4}\right\rfloor+1$ (see
Table~\ref{tab:1}).
\end{proof}
 \begin{table}[htb]
\begin{center}
\begin{tabular}{c|rrrrrrrrrrrr}
$n/k$ & $0$ & $1$ &$2$ &$3$ &$4$ &$5$ &$6$ &$7$ &$8$ &$9$ &$10$ &$11$
  \\ \hline
$n=1$ & \red{1} &&&&&&&&&&& \\
$n=2$ & \red{1} & \blue{$-1$} &&&&&&&&&& \\ 
$n=3$ & \red{1} & \blue{$-2$} & 1 &&&&&&&&& \\ 
$n=4$ & 1 & \blue{$-3$} & \red{3} & $-1$ &&&&&&&& \\ 
$n=5$ & 1 & $-4$ & \red{6} & \blue{$-4$} & 1 &&&&&&& \\ 
$n=6$ & 1 & $-5$ & \red{10} & \blue{$-10$} & 5 & $-1$ &&&&&& \\ 
$n=7$ & 1 & $-6$ & \red{15} & \blue{$-20$} & 15 & $-6$ & 1 &&&&& \\ 
$n=8$ & 1 & $-7$ & 21 & \blue{$-35$} & \red{35} & $-21$ & 7 & $-1$ &&&& \\ 
$n=9$ & 1 & $-8$ & 28 & $-56$ & \red{70} & \blue{$-56$} & 28 & $-8$ & 1 &&& \\ 
$n=10$ & 1 & $-9$ & 36 & $-84$ & \red{126} & \blue{$-126$} & 84 & $-36$ & 9 & $-1$ && \\ 
$n=11$ & 1 & $-10$ & 45 & $-120$ & \red{210} & \blue{$-252$} & 210 & $-120$ & 45 & $-10$ & 1 &\\ 
$n=12$ & 1 & $-11$ & 55 & $-105$ & 330 & \blue{$-462$} & \red{462} & $-330$ & 165 & $-55$ & 11
  & $-1$\\ 
\end{tabular}
\end{center}
\caption{Computation of the upper (red) and lower (blue) bounds. The value of
  the capacity $\mu$ is 1 for the $k+1$ first
lines of the lattice $2^N$. Each entry $(n,k)$ equals $m^\mu(N)$, as given by
(\ref{eq:boundm1}).} 
\label{tab:1}
\end{table}

\section{Exact bounds of the interaction transforms}
We begin by establishing a technical lemma which will permit to get the results
easily from Theorem~\ref{th:2.boundM}.
\begin{lemma}\label{lem:1}
Let $A,B\subset N$, $A\neq\varnothing$, be disjoint sets. Then
\begin{equation}\label{eq:max}
\max_{\mu\in\NC(N)}\sum_{C\subseteq A}(-1)^{a-c}\mu(B\cup C) = \max_{\mu\in\NC(N)}m^\mu(A),
\end{equation}
and the maximum is attained for $\mu=\mu_A^*$.
\end{lemma}
\begin{proof}
The function we have to maximize is simply the derivative $\Delta_A\mu(B)$.  As
this is a linear function in $\mu$ and $\NC(N)$ is a polytope, its maximum is
attained on a vertex, i.e. a $\{ 0, 1\} $-valued capacity. If $\mu(B\cup
A)=\mu(B)$, then by monotonicity of $\mu$ we get $\Delta_A\mu(B)=0$. Since this
is clearly not the maximum of the derivative, we can discard such capacities
$\mu$ from the analysis. Assuming then $\mu(B\cup A) >\mu(B)$, we define a
capacity $\mu_B\in\C(A)$ by
\begin{equation}\label{eq:10}
\mu_B(C) = \mu(B\cup C) - \mu(B) \qquad (C\subseteq A).
\end{equation}
Observe that if $\mu$ is $\{0,1\}$-valued, then necessarily $\mu(B\cup
  A)=1$ and $\mu(B)=0$, hence (\ref{eq:10}) collapses to $\mu_B(C) = \mu(B\cup
  C)$, for any $C\subseteq A$, and $\mu_B$ is $\{0,1\}$-valued and normalized too.
Moreover, any $\{0,1\}$-valued normalized capacity on $A$ can be obtained from a
$\{0,1\}$-valued normalized capacity on $N$ by the latter equality. 
  On the other hand, remark that for any $\mu\in \NC(N)$
\[
m^{\mu_B}(A) = \sum_{C\subseteq A}(-1)^{a-c}\mu_B(C) = \sum_{C\subseteq A}(-1)^{a-c}\mu(B\cup C)
\]
since $\sum_{C\subseteq A}(-1)^{a-c}=0$. In summary, we have
\[
\max_{\mu\in\NC(N)}\Delta_A\mu(B)=\max_{\mu\in\NC_{0,1}(N)}\Delta_A\mu(B) =
\max_{\mu\in\NC_{0,1}(A)}m^{\mu}(A) = \max_{\mu\in\NC(A)}m^{\mu}(A) = \max_{\mu\in\NC(N)}m^{\mu}(A),
\]
the last equality coming from Theorem~\ref{th:2.boundM}. Hence
(\ref{eq:max}) is established, the value of the maximum is given by
Theorem~\ref{th:2.boundM}, as well as the capacity attaining the maximum.
\end{proof}
A similar result can be established for the lower bound.
\begin{corollary}\label{cor:1}
Consider $A\subseteq N.$ The upper and lower bounds for the interaction
transform $I(A)$ are the same as
for $m(A)$, and they are obtained for the capacities $\mu^*_A$
and $\mu_{A*}$ of Theorem~\ref{th:2.boundM}.
\end{corollary}
\begin{proof}
We will obtain the upper bound, the proof for the lower bound being
similar. From Lemma~\ref{lem:1}, we see that the maximum of $\Delta_A\mu(B)$ does
not depend on $B$. Thus, from (\ref{eq:3}), letting
$m^*(A)=\max_{\mu\in\NC(N)}m^\mu(A)$ we obtain
\[
\max_{\mu\in\NC(N)}I^\mu(A) = \sum_{B\subseteq N\backslash A} \frac{(n-a-b)!b!}{(n-a+1)!} m^*(A) = m^*(A) \sum_{b=0}^{n-a} \frac{(n-a-b)!b!}{(n-a+1)!}
    \binom{n-a}{b} = m^*(A).
\]
\end{proof}

Similarly, we obtain the exact bounds for the Banzhaf interaction index.
\begin{corollary}
Consider $A\subseteq N.$ The upper and lower bounds for $I_B(A)$ are the same as
for $m(A).$ These upper and lower bounds are obtained for the capacities $\mu^*_A$
and $\mu_{A*}$ of Theorem~\ref{th:2.boundM}.
\end{corollary}
\begin{proof}
Proceeding as for Corollary~\ref{cor:1}, the result follows from the identity $\sum_{b=0}^{n-a}\binom{n-a}{b}=2^{n-a}$.
\end{proof}

\section{Exact bounds for $k$-additive and $p$-symmetric capacities}
We show in this section that the results established for the bounds of the
M\"obius and interaction transforms on the set of normalized capacities are
still valid when one restricts to $k$-additive capacities and $p$-symmetric capacities.

\begin{proposition}\label{prop:1}
For any nonempty $A\subseteq N$, the normalized capacities $\mu^*_A,\mu_{A*}$
given in Theorem~\ref{th:2.boundM} are at most $k$-additive for any
$|A|\leqslant k\leqslant n$. Therefore, the upper and lower bounds for the
M\"obius transform, the interaction transform and the Banzhaf interaction
transform, are valid:
\[
\max_{\mu\in\NC(N)}m^\mu(A) = \max_{\mu\in\NC^{\leqslant k}(N)}m^\mu(A),\quad
\min_{\mu\in\NC(N)}m^\mu(A) = \min_{\mu\in\NC^{\leqslant k}(N)}m^\mu(A),
\]
for $|A|\leqslant k\leqslant n,\varnothing\neq A\subseteq N$,
and similarly for $I^\mu(A),I^\mu_\mathrm{B}(A)$.
\end{proposition}
\begin{proof}
Given a nonempty $A\subseteq N$, it suffices to show that $ \mu_A^{*}, \mu_{A*}$ are at
most $k$-additive for $k=|A|$. Take $B\subseteq N$ such that $k<|B|\leqslant n.$
Then, $B\setminus A\neq\varnothing$. On the other hand, observe that for
  any $i\not\in A$, 
\[
\Delta_i\mu^*_A(K)= \mu^*_A(K\cup i) -\mu^*_A(K)=0
\]
for any $K\not\ni i$. It follows that $\Delta_B\mu^*_A(K)=0$ for any $K$ as soon as
$B\setminus A\neq\varnothing$. Taking $K=\varnothing$, by (\ref{eq:md}), we conclude that
$m^{\mu^*_A}(B) = 0$ if $k<|B|\leqslant n$, as desired.
\end{proof}
 \begin{remark} Proposition~\ref{prop:1} tells us what is the maximum
    achieved by $m^\mu(A)$ for the set of $k$-additive capacities when
    $|A|\leqslant k\leqslant n$, but says nothing when $k<|A|$. The question
    appears to be very complex, because in general $\mu^*_A$ will not be
    $k$-additive, and the vertices of the polytope of $k$-additive capacities
    are not known, except for $k=1$ and 2. In particular, it is known that many
    vertices are \textit{not} $\{0,1\}$-valued as soon as $k>2$ (see
    \cite{micogi06}).
\end{remark}
\begin{proposition}
For any $1\leqslant p\leqslant n$ and any partition $\{ A_1,\ldots, A_p\} $ of $N$,
\[
\max_{\mu\in\NC(N)}m^\mu(A) = \max_{\mu\in \SNC^{\leqslant p}(A_1,\ldots, A_p)}m^\mu(A), \qquad
(\varnothing\neq A\subseteq N), \]
\[ \min_{\mu\in\NC(N)}m^\mu(A) = \min_{\mu\in \SNC^{\leqslant p}(A_1,\ldots, A_p)}m^\mu(A), \qquad
(\varnothing\neq A\subseteq N),
\]
and similarly for $I^\mu(A),I^\mu_\mathrm{B}(A)$.
\end{proposition}
\begin{proof}
Consider the capacities defined by
\[
\mu_A^{**} (B):=\left\{ \begin{array}{cc} 1 & \mbox{~if~} |B| \geq l_{|A|}+1
  \\ 0 & \mbox{~otherwise} \end{array}\right. ,\, \, \, \mu_{A**}
(B):=\left\{ \begin{array}{cc} 1 & \mbox{~if~} |B| \geq l_{|A|} \\ 0 &
  \mbox{~otherwise} \end{array}\right.
\]
 Observe that
$\mu_A^{**}(C)=\mu_A^{*}(C),\, \mu_{A**} (C)=\mu_{A*} (C)$ for any $C\subseteq
A.$

Therefore $m^{\mu^{**}_A}(A)=m^{\mu^*_A}(A)$, $m^{\mu_{A**}}(A)=m^{\mu_{A*}}(A).$
On the other hand, $\mu_A^{**}$ and $\mu_{A**}$ are symmetric capacities, whence
they are $p$-symmetric for any $p$ and any partition of indifference.
\end{proof}

\section{{\bf Acknowledgements}}

We address all our thanks to the anonymous referees, for their constructive
comments, which have permitted to correct and improve the paper, in particular
the proof of Theorem~1 which is now much shorter and more clear.  This work was
partially supported by Grant MTM2012-33740.

\bibliographystyle{plain}
\bibliography{../BIB/fuzzy,../BIB/grabisch,../BIB/general}

\end{document}